%% file: main.tex
\documentclass[letterpaper, 10 pt, conference]{ieeeconf}  

\IEEEoverridecommandlockouts                              
\input{definition}

\title{\LARGE \bf
Coordinated Path Following of UAVs using Event-Triggered Communication over Networks with Digraph Topologies
}

\author{Hyungsoo Kang, Isaac Kaminer, Venanzio Cichella, and Naira Hovakimyan
\thanks{This work is supported by Air Force Office of Scientific Research (AFOSR) grant FA9550-21-1-0411, National Aeronautics and Space Administration (NASA) University Leadership Initiative (ULI) grant 80NSSC17M0051, and ONR Science of Autonomy Program under grant N0001424WX01651.}
\thanks{Hyungsoo Kang and Naira Hovakimyan are with the Department of Mechanical Science and Engineering, University of Illinois at Urbana-Champaign, 
        Urbana, IL 61801, USA.  
        {\tt\small \{hk15, nhovakim\} @illinois.edu}}
\thanks{Isaac Kaminer is with the Department of
Mechanical and Aerospace Engineering, Naval Postgraduate School, Monterey, CA 93943, USA.
        {\tt\small kaminer@nps.edu}}
\thanks{Venanzio Cichella is with the Department of Mechanical
Engineering, University of Iowa, 
        Iowa City, IA 52242, USA.
        {\tt\small venanzio-cichella@uiowa.edu}} 
}

\begin{document}

\maketitle
\thispagestyle{empty}
\pagestyle{empty}

\begin{abstract}
This article presents a novel time-coordination algorithm based on event-triggered communication to ensure multiple UAVs progress along their desired paths in coordination with one another. In the proposed algorithm, a UAV transmits its progression information to its neighbor UAVs only when a decentralized trigger condition is satisfied. Consequently, it significantly reduces the volume of inter-vehicle communications required to achieve the goal compared with the existing algorithms based on continuous communication. With such intermittent communications, it is shown that a decentralized coordination controller guarantees exponential convergence of the coordination error to a neighborhood of zero. Furthermore, a lower bound on the difference between two consecutive event-triggered times is provided showing that the Zeno behavior is excluded with the proposed algorithm. Lastly, simulation results validate the efficacy of the proposed algorithm.
\end{abstract}

\section{INTRODUCTION}
Recent remarkable progress in the theory and technologies of multi-agent systems has given rise to a widespread use of multi-UAV systems. Relevant applications include collaborative payload transportation \cite{Lee2013,Lee2017}, formation flying for space exploration \cite{Morgan20141725,Bandyopadhyay2015}, and cooperative SLAM \cite{Schmuck2017,Dubé2017} to name a few.

Among the diverse algorithms, coordinated path-following control has played a key role in solving challenging problems such as 1) search and rescue missions, where multiple UAVs progress in coordination covering a large area at a time; 2) sequential safe auto-landing, where multiple UAVs arrive at the glide slope safely separated by a predefined time interval; 3) atmospheric-science missions, where multiple UAVs fly cooperatively collecting data from a region of interest. 
The key requirement that the coordinated path following framework was developed to address was to guarantee a simultaneous arrival of each UAV at the end of its desired trajectory. The framework includes three steps: 1)  Generate a set of collision-free desired trajectories to be assigned to each UAV that minimize a given cost and satisfy the simultaneous arrival requirement, boundary conditions, UAV dynamic constraints, and guarantee obstacle avoidance \cite{Choe20161744}, \cite{Cichella2021}. 
2) Design  a  path-following control law \cite{CICHELLA201313} that steers a UAV along its desired trajectory.
3) Develop a decentralized time-coordination algorithm that enables each UAV to transmit progression information along its desired trajectory to its neighbors and also to adjust its progression speed based on the information provided by the neighbors. This step guarantees simultaneous time arrival by all the UAVs in the presence of disturbances.  

It has been shown in the early-stage research on coordinated path following that in Step 3 above the exponential stability of the time-coordination algorithms employed by all the UAVs can be reduced to a consensus problem. Moreover, in, for example, \cite{Kaminer2006}, \cite{Ghabcheloo2007133} it was assumed that the topology of the underlying communication network  is represented by a connected  bidirectional graph with a fixed topology. However, this is a strong assumption less likely to be satisfied by a typical communication network with
time-varying network topology. To address this issue,  researchers in \cite{Xargay2013499}, \cite{Cichella2015945} were able to guarantee convergence of the time-coordination algorithms for the case where the network topology is represented by a time-varying bidirectional graph that is connected in an integral sense, i.e. the integral of the graph from $t$ to $t+T$ is connected $\forall t\geq0$ with a constant $T>0$. This condition, connectedness in an integral sense, was used in \cite{Mehdi2017}, \cite{Tabasso2020436} to show that collision avoidance can be achieved as well as the time-coordination. In the work reported in \cite{Tabasso2022704}, the researchers applied the time-coordination algorithm \cite{Cichella2015945} to the problem of continuous monitoring of a path-constrained moving target. Our recent work \cite{Kang2024} showed convergence of time-coordination algorithms for a more general connectivity condition: the communication network is no longer required to be bidirectional. In fact, we have shown that the connectedness of the directed graph in an integral sense is sufficient to achieve  convergence. We note, however, that all the  algorithms discussed rely on continuous or piecewise continuous communications which may be undesirable or unavailable in real-world applications. Typical examples include underwater robotic missions where communication bandwidth is severely limited and military applications where stealth is of paramount importance. 

This issue can be resolved by resorting to event-triggered communication and control (ETC) algorithms. With ETC-based methods, the communication and control input updates occur only when a predefined condition is satisfied, thereby significantly reducing the inter-vehicle communications in the process of achieving consensus on variables of interest. The ETC-based coordination algorithms were pioneered by \cite{Dimarogonas2012,SEYBOTH2013245} with single-integrator dynamics. The algorithms have been improved to achieve consensus of the state variables of linear dynamics \cite{YANG2019129,Hu2018,Qian2019,Wu2018,Cheng2019} and nonlinear dynamics \cite{Nguyen2022,ZHANG2021,Wang2022}. In respect of the topology of the communication network, most of the early works assumed it to be a static bidirectional graph \cite{Dimarogonas2012,SEYBOTH2013245,Fan2015,YANG2019129} or digraph \cite{Hu2018},~\cite{Qian2019}. Recently, some algorithms have been proposed where the network topology is assumed to be a time-varying bidirectional graph \cite{Wu2018,Cheng2019,Hu2019} or digraph \cite{Jia2018,HAN2015196,Hao2023}. 
These existing event-triggered communications (ETC) based algorithms have been developed to attain the state consensus $\gamma_i(t)=\gamma_j(t), \ \forall i,j$. Importantly, these approaches cannot be applied to the UAV time coordination problem because it has more ambitious objectives: achieving the state consensus $\gamma_i(t)=\gamma_j(t), \ \forall i,j$ under the requirement that each $\dot{\gamma}_i(t)$ tracks a given reference signal $\dot{\gamma}_d(t)$. The additional requirement on $\dot{\gamma}_i(t)$ makes the consensus problem trickier. Thus, it is natural to ask whether we can design an ETC-based algorithm to achieve the ambitious time-coordination objectives. This paper affirmatively answers it. 

The contributions of this paper are summarized next. We propose a novel ETC-based consensus algorithm that attains the time-coordination objectives: the state consensus $\gamma_i(t)=\gamma_j(t), \ \forall i,j$ with each $\dot{\gamma}_i(t)$ tracking a given reference signal $\dot{\gamma}_d(t)$.
With the presented algorithm, a UAV transmits its progression information only when an event-triggering condition is satisfied, thereby significantly reducing the inter-vehicle communication requirement, particularly when compared with the existing time-coordination algorithms based on continuous communications, for example \cite{Kaminer2006,Ghabcheloo2007133,Xargay2013499,Cichella2015945}. We employ Lyapunov analysis to show that the proposed ETC based distributed control law achieves the coordination objectives with exponential stability. Also, it is proven that the time interval between two consecutive events is bounded from below. This implies that Zeno behavior is excluded with the proposed algorithm. Since the algorithm is designed without making assumptions on the vehicle dynamics, it is applicable to any vehicle endowed with a path-following controller. In this paper, we illustrate the use of the algorithm  on quadrotors, a very widely used UAV. 

The rest of this article is organized as follows. Section~\ref{prelim} provides a brief review of graph theory. Section~\ref{III} formulates the time-coordination problem and presents assumptions on the inter-vehicle information flow. In Section~\ref{IV}, the time-coordination control law based on event-triggered communication is described. The performance of the time-coordination algorithm is analyzed in the main theorem. In Section~\ref{V}, simulation results validate the efficacy of the proposed algorithm. Finally, Section~\ref{VI} presents some conclusions.

\section{PRELIMINARIES} \label{prelim}
\subsection{Graph Theory}

\noindent


A digraph of size $n$ is defined by $\mathcal{D}=(\mathcal{V},\mathcal{E},\mathcal{A})$, where $\mathcal{V}=\{1,\dots,n\}$ is the set of nodes, $\mathcal{E}$ is the set of edges, and $\mathcal{A}$ is the adjacency matrix. An edge is denoted by an ordered pair $(i,j)$, which means information can be transmitted from node $j$ to node $i$. The adjacency matrix $\mathcal{A}$ is constructed as follows: if $(i,j)\in \mathcal{E}$, one has $\mathcal{A}_{ij}=1$. Otherwise, $\mathcal{A}_{ij}=0$. The digraph $\mathcal{D}$ is represented by the Laplacian $L\trieq \Delta-\mathcal{A}$, where $\Delta$ is a diagonal matrix with $\Delta_{ii}\trieq\sum_{j=1,j\neq i}^{n}\mathcal{A}_{ij}$. The neighborhood of node $i$ is the set $\mathcal{N}_i\trieq\{j\in\mathcal{V}: (i,j)\in \mathcal{E}\}$. A directed path from node $i_s$ to node $i_0$ is a sequence of edges $(i_0,i_1)$, $(i_1,i_2)$, $\dots$, $(i_{s-1},i_s)$. The digraph $\mathcal{D}$ is said to be connected via a directed spanning tree if there exists at least one node from which every other node in the digraph can be reached via a directed path. In this case, the Laplacian $L$ has one eigenvalue at $0$ and all the rest are in the right half plane.


\section{TIME-COORDINATED PATH-FOLLOWING FRAMEWORK} \label{III}
\subsection{Path Following of a Single UAV}
A trajectory generation algorithm such as \cite{Choe20161744}, \cite{Cichella2021} produces a set of collision-free desired trajectories for $n$ UAVs
\begin{align} \label{trajectory}
    p_{d,i}(t_d): [0,t_f]\rightarrow \mathbb{R}^3, \ \ i\in\{1,\dots,n\},
\end{align}
where $t_f$ is the simultaneous time of arrival. Let us introduce a tunable variable $\gamma_i(t)$ called coordination state or virtual time
\begin{align*}
    \gamma_i(t): [0,\infty)\rightarrow[0,t_f], \ \ \ i\in\{1,\dots,n\}. 
\end{align*}
Notice that the range of $\gamma_i(t)$ is the same as the domain of the trajectory $p_{d,i}(\cdot)$. Consequently, it is possible to construct a composite function $p_{d,i}(\gamma_i(t))$, which can be at any point on the trajectory $p_{d,i}(\cdot)$ depending on the value of $\gamma_i(t)$. We use the expression $p_{d,i}(\gamma_i(t))$ to describe the desired position of the $i$th UAV, called the virtual target. 
The coordination state $\gamma_i(t)$ is an additional degree of freedom that allows one to adjust the progression of the UAV along its trajectory by controlling its value. This paper considers a case where the inter-vehicle coordination is destroyed by disturbances in the middle of the mission. We design a controller for $\gamma_i(t)$ such that it prevents UAVs from being put outside of the finite region of attraction of a path-following controller and recovers the inter-vehicle coordination automatically.

The path-following controller \cite{CICHELLA201313} gets the UAV to track $p_{d,i}(\gamma_i(t))$ by ensuring that the path following error 
\begin{align*}
    e_{PF,i}(t)\trieq p_i(t)-p_{d,i}(\gamma_i(t)), \ \ \ i\in\{1,\dots,n\}, 
\end{align*}
converges to zero. In the above expression, $p_i(t)$ represents the actual position of the UAV. The control law ensures exponential convergence to zero with the ideal performance of the inner-loop autopilot and to a neighborhood of zero with a non-ideal one. In other words, in the latter case, there exists $\rho>0$ such that 
\begin{align} \label{pf_error}
    \|e_{PF}(t)\|\leq\rho, \ \ \ \forall t\geq0,
\end{align}
where $e_{PF}(t)=[e_{PF,1}(t)^\top,\dots,e_{PF,n}(t)^\top]^\top$.

\subsection{Time Coordination of Multiple UAVs}
As the position of the virtual target is $p_{d,i}(\gamma_i(t))$ and the velocity of the virtual target is $\frac{dp_{d,i}(\gamma_i(t))}{dt}=\frac{dp_{d,i}(\gamma_i(t))}{d\gamma_i(t)}\dot{\gamma}_i(t)$, we can notice that $\gamma_i(t)$ and $\dot{\gamma}_i(t)$ characterize the progression of the UAV. Based on this, we formulate the following objectives. 

The UAVs are said to be synchronized at time $t$, if 
\begin{align} \label{obj1}
    \gamma_i(t)=\gamma_j(t), \ \ \ \forall i,j\in\{1,\dots,n\}.
\end{align}
Furthermore, for a desired mission progression pace $\dot{\gamma}_d(t)>0$, if 
\begin{align} \label{obj2}
    \dot{\gamma}_i(t)=\dot{\gamma}_d(t), \ \ \ \forall i\in\{1,\dots,n\},
\end{align}
the UAVs are considered to be progressing in accordance with the desired progression mission pace. 

To satisfy the above requirements, UAVs need to exchange their coordination states $\gamma_i(t)$ with their neighbors and tune the values.
This interaction among the UAVs can be modeled by using graph theory, see Section \ref{prelim}. 

The assumptions on the nature of the inter-vehicle communications are provided next.

\begin{assumption} \label{assum1}
The information flow between any two UAVs is directional without time delays.
\end{assumption}

\begin{assumption} \label{assum2}
The $i$th UAV can receive coordination information $\gamma_j(t)$ only 
from UAVs in its neighborhood set $\mathcal{N}_i$, where $j\in\mathcal{N}_i$.
\end{assumption}

The topology of the underlying communication network $\mathcal{D}$ 
satisfies the following assumption.

\begin{assumption} \label{assum3}

The communication network modeled by a digraph $\mathcal{D}$ contains a directed spanning tree. In other words, a root node in it can reach every other node via a directed path. 
\end{assumption}

\begin{remark}
    While Assumption~\ref{assum3} posits that the communication network contains a directed spanning tree—thereby suggesting the potential for continuous communication among nodes (or vehicles)—it is important to underscore that continuous communication is not required. Instead, the UAVs are designed to engage in communication only at specific, discrete events. Although the network is capable of facilitating uninterrupted communication, in practice, UAVs will only communicate when necessary, at predetermined events. Hence, the availability of communication is guaranteed at these events, ensuring that the system functions efficiently without the need for constant communication.
\end{remark}



\noindent \textit{Problem (Time-Coordinated Path-Following Problem):} Given a set of desired trajectories \eqref{trajectory} and a path-following control law that ensures \eqref{pf_error}, design a decentralized time-coordination control law such that the coordination states $\gamma_i(t)$  converge for all $i$ exponentially to a neighborhood of the equilibrium \eqref{obj1} and \eqref{obj2} under Assumptions~\ref{assum1}, \ref{assum2}, and \ref{assum3}.

\begin{remark}
    Designing a time-coordination control law for the coordination states $\gamma_i(t)$ that solves the Time-Coordinated Path-Following Problem guarantees the simultaneous arrival of all the UAVs at their respective destinations.
\end{remark}
\section{MAIN RESULT} \label{IV}
In this section, a decentralized time-coordination algorithm based on event-triggered communication (ETC) is presented. 

We propose the following decentralized control law under Assumptions \ref{assum1} and \ref{assum2}
\begin{align} \label{dyn1}
    \ddot{\gamma}_i(t)&=-b(\dot{\gamma}_i(t)-\dot{\gamma}_d(t)) \nonumber \\
    &\mathrel{\phantom{=}}-a\sum_{j\in\mathcal{N}_i}(\gamma_i(t)-\hat{\gamma}_j(t))+\bar{\alpha}_i(e_{PF,i}(t)), \\
    \gamma_i(0)&=\gamma_{i0}, \ \ \dot{\gamma}_i(0)=\dot{\gamma}_{i0}, \nonumber
\end{align}
where $a$ and $b$ are positive coordination control gains and $\bar{\alpha}_i(e_{PF,i}(t))$ is defined as 
\begin{align} \label{alpha}
    \bar{\alpha}_i(e_{PF,i}(t))=k_{PF}\frac{\dot{p}_{d,i}(\gamma_i(t))^\top e_{PF,i}(t)}{\|\dot{p}_{d,i}(\gamma_i(t))\|+\eta}
\end{align}
with $k_{PF}$ and $\eta$ being positive design parameters.

When the UAV is deviated from its desired position due to disturbances, the term $\bar{\alpha}_i(e_{PF,i}(t))$ adjusts $\gamma_i(t)$ in a way that the virtual target $p_{d,i}(\gamma_i(t))$ moves toward the UAV securing it inside the finite region of attraction of the path-following controller, which leads to discoordination of the fleet of UAVs. 
This issue is resolved by virtue of the second term in \eqref{dyn1}, which is geared towards recovering the inter-vehicle coordination.
Therefore, even if the UAVs are hit by disturbances and discoordinated in the middle of the mission, the controller \eqref{dyn1} helps to secure the UAVs and recover the coordination, leading to simultaneous arrival at the final destinations as planned in \eqref{trajectory}. The first term in \eqref{dyn1} is designed for $\dot{\gamma}_i(t)$ to track $\dot{\gamma}_d(t)$. Considering that the velocity of the virtual target is $\frac{dp_{d,i}(\gamma_i(t))}{dt}=\frac{dp_{d,i}(\gamma_i(t))}{d\gamma_i(t)}\dot{\gamma}_i(t)$, it is used for tuning the speed of the fleet, e.g. early arrival with $\dot{\gamma}_d(t)>1$ and late arrival with $\dot{\gamma}_d(t)<1$.    

Now, let us go into our ETC algorithm. Under the ETC framework, the $j$th UAV transmits to the $i$th UAV its coordination information $\gamma_j(t)$ only when its triggering event occurs. Since the $i$th UAV receives $\gamma_j(t)$ intermittently, it is reasonable for it to use an estimate $\hat{\gamma}_j(t)$ to adjust its own state $\gamma_i(t)$ as in \eqref{dyn1} between the event time instants of the $j$th UAV. The estimator of $\gamma_j(t)$ is given below: 
\begin{equation} \label{estimator}
\hat{\gamma}_j(t):\left\{ \begin{aligned} 
  &\ddot{\hat{\gamma}}_j(t)=-b(\dot{\hat{\gamma}}_j(t)-\dot{\gamma}_d(t)), \ \ t\in[t^j_{k}, \ t^j_{k+1}), \\
  &\dot{\hat{\gamma}}_j(t^j_{k})=\dot{\gamma}_j(t^j_{k}), \ \ \hat{\gamma}_j(t^j_{k})=\gamma_j(t^j_{k})
\end{aligned} \right.,
\end{equation}
where $t^j_{k}$ denotes the $k$th event time of the $j$th UAV at which it transmits its coordination state ${\gamma}_j(t^j_{k})$ and rate $\dot{\gamma}_j(t^j_{k})$.

Let the estimation error be $e_j(t)\trieq\hat{\gamma}_j(t)-\gamma_j(t)$. Then the controller \eqref{dyn1} can be rewritten as follows
\begin{align} \label{dyn1'}
    \ddot{\gamma}_i(t)&=-b(\dot{\gamma}_i(t)-\dot{\gamma}_d(t))-a\sum_{j\in\mathcal{N}_i}(\gamma_i(t)-\gamma_j(t)) \nonumber \\
    &\mathrel{\phantom{=}}+a\sum_{j\in\mathcal{N}_i}e_j(t)+\bar{\alpha}_i(e_{PF,i}(t)), \\
    \gamma_i(0)&=\gamma_{i0}, \ \ \dot{\gamma}_i(0)=\dot{\gamma}_{i0}. \nonumber
\end{align}
As it can be seen in \eqref{dyn1'}, the error $e_j(t)$ has an explicit impact on the time-coordination dynamics. The key point of the ETC algorithm is to ensure that $|e_j(t)|$ remains bounded. Therefore, when $|e_j(t)|$ reaches a given threshold, i.e., a transmission event is triggered, the $j$th UAV  transmits $\gamma_j(t)$ and $\dot{\gamma}_j(t)$. To be more specific, we define an event-triggering function $\delta_j(t)$ as
\begin{align}
    \delta_j(t)=|e_j(t)|-h(t),
\end{align}
where $h(t)=c_1+c_2e^{-c_3 t}$, $c_1>0$, $c_2,c_3\geq0$, is referred to as a threshold function, and $c_1\leq h(t)\leq c_1+c_2$. Whenever $\delta_j(t)> 0$, sampling and transmission take~place with $|e_j(t)|$ reset to $0$. With this ETC algorithm, boundedness of the error is guaranteed: $|e_j(t)|\leq h(t)$.
\begin{remark}
    The $i$th UAV uses the estimator \eqref{estimator} to propagate the time-coordination controller \eqref{dyn1}. On the other hand, the $j$th UAV executes the same code to check whether the event triggering condition $|e_j(t)|>h(t)$ is satisfied.
\end{remark}
\begin{remark}
    Comparing the estimator \eqref{estimator} with the controller \eqref{dyn1}, we can notice that the error $e_j(t)\trieq\hat{\gamma}_j(t)-\gamma_j(t)$ stems from $a\sum_{k\in\mathcal{N}_j}(\gamma_j(t)-\hat{\gamma}_k(t))+\bar{\alpha}_j(e_{PF,j}(t))$. When the $j$th UAV is pushed backward by headwinds, $\bar{\alpha}_j(e_{PF,j}(t))$ decreases fast from zero causing deceleration of $\gamma_j(t)$. It leads to a fast evolution of $a\sum_{k\in\mathcal{N}_j}(\gamma_j(t)-\hat{\gamma}_k(t))+\bar{\alpha}_j(e_{PF,j}(t))$. Consequently, the error $e_j(t)$ frequently reaches a threshold value triggering the event and transmission of the signal. In the other case when the $j$th UAV is in coordination with the other UAVs, $a\sum_{k\in\mathcal{N}_j}(\gamma_j(t)-\hat{\gamma}_k(t))+\bar{\alpha}_j(e_{PF,j}(t))$ is close to zero and the event is not triggered stopping transmitting the signal. To put it briefly, when the UAV falls behind the virtual target, it sends a message to other UAVs requesting a slowdown of their speeds.
\end{remark}
For the ease of stability analysis, we introduce the coordination error state $\xi_{TC}(t)=[\xi_1(t)^\top \ \xi_2(t)^\top]^\top$ with
\begin{equation}
\begin{aligned} \label{error}
    \xi_1(t)&=Q\gamma(t) \ \in \mathbb{R}^{n-1}, \\
    \xi_2(t)&=\dot{\gamma}(t)-\dot{\gamma}_d(t)1_n \ \in \mathbb{R}^{n},
\end{aligned}
\end{equation}
where $\gamma(t)=[\gamma_1(t), \dots, \gamma_n(t)]^\top$, and $Q\in\mathbb{R}^{(n-1)\times n}$ is a matrix that satisfies $Q1_n=0_{n-1}$ and $QQ^\top=\mathbb{I}_{n-1}$. 
\begin{remark}
    A matrix $Q_k\in\mathbb{R}^{(k-1)\times k}$, $k\geq2$, satisfying $Q_k1_k=0_{k-1}$ and $Q_k\left(Q_k\right)^\top=\mathbb{I}_{k-1}$, can be constructed recursively:
    \begin{align*}
        Q_k=
        \begin{bmatrix}
        \sqrt{\frac{k-1}{k}} & -\frac{1}{\sqrt{k(k-1)}}1_{k-1}^\top \\
        0 & Q_{k-1} \\
        \end{bmatrix}
    \end{align*}
    with initial condition $Q_2=[1/\sqrt{2} \ -1/\sqrt{2}]$. For notational simplicity, we denote $Q_n$ by $Q$, where $n$ is the number of UAVs.
\end{remark}
It is shown in \cite[Lemma~7]{phdenric2013} that $Q^\top Q=\mathbb{I}_n-\frac{1_n1_n^\top}{n}$ and the nullspace of Q is spanned by $1_n$. Now suppose that $\xi_1(t)=Q\gamma(t)=0_{n-1}$. Then since the nullspace of Q is spanned by $1_n$, we obtain that $\gamma_i(t)=\gamma_j(t)$, $\forall i,j\in\{1,\dots,n\}$. Furthermore, $\xi_2(t)=0_n$ implies that 
$\dot{\gamma}_i(t)=\dot{\gamma}_d(t)$, $\forall i\in\{1,\dots,n\}$. Therefore, $\xi_{TC}(t)=0_{2n-1}$ is equivalent to \eqref{obj1} and \eqref{obj2}. 

\begin{lemma} \label{lem1}
    The spectrum of $\bar{L}\trieq QLQ^\top$ is the same as that of $L$ without the first eigenvalue $\lambda_1=0$.
\end{lemma}
\begin{proof}
    From $LQ^\top Q=L\left(\mathbb{I}_n-\frac{1_n1_n^\top}{n}\right)=L-0=L$, $Lx=\lambda x$ is rewritten as $LQ^\top Q x=\lambda x$.~Premultiplying by $Q$ gives $QLQ^\top Q x=\lambda Qx$, which is equivalent to $\bar{L}(Qx)=\lambda(Qx)$. Additionally, since $1_n$ is in the nullspace of $Q$, the eigenvalue $\lambda_1=0$ of $L$ with its corresponding $1_n$ cannot be included in the spectrum of $\bar{L}$.
\end{proof}

Using the definitions of $\gamma(t)$, the Laplacian $L$ and the~adjacency matrix $\mathcal{A}$, the expression \eqref{dyn1'} can be rewritten in a more compact form
\begin{align} 
    \ddot{\gamma}(t)&=-b\xi_2(t)-aL\gamma(t)+a\mathcal{A}e(t)+\bar{\alpha}(e_{PF}(t)), \nonumber\\
    \gamma(0)&=\gamma_0, \ \ \dot{\gamma}(0)=\dot{\gamma}_0, \nonumber
\end{align} 
where $\bar{\alpha}(e_{PF}$ $(t))=[\bar{\alpha}_1(e_{PF,1}(t)),\dots,\bar{\alpha}_n(e_{PF,n}(t))]^\top$ and $e(t)=[e_1(t),\dots,e_n(t)]^\top$.

The main results of the paper are presented in the following theorem.
\begin{theorem} \label{thm}
Consider a set of desired trajectories \eqref{trajectory} and a path-following controller that ensures \eqref{pf_error}. Let the evolution of $\gamma_i(t)$ be governed by \eqref{dyn1} over the network $\mathcal{D}$ satisfying Assumption~\ref{assum3}. Then there exist time coordination control gains $a$, $b$, and $\eta$ such that
\begin{equation}
\begin{aligned}
    \|\xi_{TC}(t)\|&\leq\kappa_1\|\xi_{TC}(0)\|e^{-\lambda_{TC}t} \\&\mathrel{\phantom{=}}+\kappa_2\sup_{t\geq0}\left(an\sqrt{n}h(t)+k_{PF}\|e_{PF}(t)\|+|\ddot{\gamma}_d(t)|\right), \nonumber
\end{aligned}
\end{equation}
where $\lambda_{TC}$, $\kappa_1$, and $\kappa_2$ are presented in the proof. 

Moreover, the event-triggered time interval $t^i_{k+1}-t^i_k$ is  bounded from below. 
\end{theorem}
\begin{proof}
Motivated by \cite{Cichella2015945}, we introduce a new state
\begin{align*}
    \chi(t)=b\xi_1(t)+Q\xi_2(t).
\end{align*}
Then the coordination error state $\xi_{TC}(t)=[\xi_1(t)^\top \ \xi_2(t)^\top]^\top$ can be redefined as $\bar{\xi}_{TC}(t)=[\chi(t)^\top \ \xi_2(t)^\top]^\top$ with dynamics
\begin{equation}
\begin{aligned} \label{dyn3}
    \dot{\chi}&=-\frac{a}{b}\bar{L}\chi+\frac{a}{b}QL\xi_2+aQ\mathcal{A}e+Q\bar{\alpha}(e_{PF}) \\
    \dot{\xi}_2&=-\frac{a}{b}LQ^\top\chi-\left(b\mathbb{I}_n-\frac{a}{b}L\right)\xi_2+a\mathcal{A}e +\bar{\alpha}(e_{PF})-\ddot{\gamma}_d1_n,
\end{aligned}
\end{equation}
where $\bar{L}\trieq QLQ^\top \in \mathbb{R}^{(n-1)\times(n-1)}$. 
\\
In order to construct a Lyapunov function candidate for~\eqref{dyn3}, we first show that $\dot{\phi}=-\bar{L}\phi$ is asymptotically stable. Under Assumption~\ref{assum3}, the algebraic multiplicity of the first eigenvalue $\lambda_1=0$ of $L$ is one and the remaining eigenvalues are in the open right half plane. Applying Lemma~\ref{lem1} to this fact implies that all the eigenvalues of $-\bar{L}$ are in the open left half plane; thus $\dot{\phi}=-\bar{L}\phi$ is asymptotically stable. According to the Lyapunov equation, for a given $\Xi=\Xi^\top>0$, there exists $\Psi=\Psi^\top>0$ such that
\begin{align} \label{Lyapunov_equ}
    -\bar{L}^\top\Psi-\Psi\bar{L}=-\Xi
\end{align}
Using $\Psi$, we define a Lyapunov function candidate for \eqref{dyn3} as follows:
\begin{align} \label{lya}
    V_{TC}=\chi^\top \Psi\chi+\frac{\beta}{2}\|\xi_2\|^2=\bar{\xi}^\top_{TC} W\bar{\xi}_{TC},
\end{align}
where $\beta>0$ and $W\trieq\begin{bmatrix}
\Psi & 0 \\ 
0 & \frac{\beta}{2}\mathbb{I}_n
\end{bmatrix}$. Notice that $V_{TC}$ satisfies $z^\top M_1z\leq V_{TC}\leq z^\top M_2z$, where $z\trieq[\|\chi\| \ \|\xi_2\|]^\top$ with 
\begin{align*} 
M_1\trieq\begin{bmatrix}
\lambda_{min}(\Psi) & 0 \\ 
0 & \beta/2
\end{bmatrix} \text{ and }
M_2\trieq\begin{bmatrix}
\lambda_{max}(\Psi) & 0 \\ 
0 & \beta/2
\end{bmatrix}.
\end{align*}
The time derivative of \eqref{lya} along the trajectory of \eqref{dyn3} satisfies
\begin{align*}
    \dot{V}_{TC}&=\chi^\top\left(-\frac{a}{b}\bar{L}^\top \Psi-\frac{a}{b}\Psi\bar{L}\right)\chi-\beta\xi^\top_2\left(b\mathbb{I}_n-\frac{a}{b}L\right)\xi_2 \\
    &\mathrel{\phantom{=}}+\chi^\top\left(2\frac{a}{b}\Psi QL-\beta\frac{a}{b}QL^\top\right)\xi_2 \\
    &\mathrel{\phantom{=}}+\left(2\chi^\top \Psi Q+\beta\xi^\top_2\right)\left(a\mathcal{A}e+\bar{\alpha}(e_{PF})\right)-\beta\xi^\top_2\ddot{\gamma}_d1_n,
\end{align*}
\noindent which leads to
\begin{align*}
    \dot{V}_{TC}\leq& -\frac{a}{b}\chi^\top\Xi\chi-\beta\left(b-\frac{a}{b}n\right)\|\xi_2\|^2 \\
    &+\left(2\frac{a}{b}n\|\Psi\|+\beta\frac{a}{b}n\right)\|\chi\|\|\xi_2\| \\
    &+\left(2\|\Psi\|\|\chi\|+\beta\|\xi_2\|\right)\left(an\|e\|+\|\bar{\alpha}(e_{PF})\|+|\ddot{\gamma}_d|\right),
\end{align*}
where we used \eqref{Lyapunov_equ}, $\|Q\|=1$, $\|L\|\leq n$, and $\|\mathcal{A}\|\leq n$. \\
\begin{adjustbox}{max width=\linewidth}
\parbox{\linewidth}{\begin{align*}
    \dot{V}_{TC}\leq& -\frac{a}{b}\lambda_{min}(\Xi)\|\chi\|^2-\beta\left(b-\frac{a}{b}n\right)\|\xi_2\|^2 \\
    &+\left(2\frac{a}{b}n\lambda_{max}(\Psi)+\beta\frac{a}{b}n\right)\|\chi\|\|\xi_2\| \\
    &+\left(2\lambda_{max}(\Psi)+\beta\right)\|\bar{\xi}_{TC}\|\left(an\|e\|+k_{PF}\frac{v_{max}}{v_{min}+\eta}\|e_{PF}\|+|\ddot{\gamma}_d|\right),
\end{align*}}
\end{adjustbox}
where $v_{max}=\max_{i}\{v_{i,max}\}$ and $v_{min}=\max_{i}\{v_{i,min}\}$ with $v_{i,max}$ and $v_{i,min}$ being the maximum and minimum achievable speed of the $i$th UAV.
Letting $\eta>v_{max}-v_{min}$, one obtains

\begin{adjustbox}{max width=\linewidth}
\parbox{\linewidth}{\begin{align*}
    \dot{V}_{TC}\leq &-z^\top Uz \\
    &+\left(2\lambda_{max}(\Psi)+\beta\right)\|\bar{\xi}_{TC}\|\left(an\|e\|+k_{PF}\|e_{PF}\|+|\ddot{\gamma}_d|\right),
\end{align*}}
\end{adjustbox}
where $z=[\|\chi\| \ \|\xi_2\|]^\top$ and \\
\begin{adjustbox}{max width=\linewidth}
\parbox{\linewidth}{\begin{align*}
U\trieq\begin{bmatrix}
\frac{a}{b}\lambda_{min}(\Xi) & -\frac{1}{2}\left(\frac{2a}{b}n\lambda_{max}(\Psi)+\beta\frac{a}{b}n\right) \\ 
-\frac{1}{2}\left(\frac{2a}{b}n\lambda_{max}(\Psi)+\beta\frac{a}{b}n\right) & \beta\left(b-\frac{a}{b}n\right)
\end{bmatrix}.
\end{align*}}
\end{adjustbox}
Next, we introduce
\begin{align} \label{convergence rate}
    \lambda_{TC}\leq\frac{a}{b}\frac{\lambda_{min}(\Xi)}{3\lambda_{max}(\Psi)}, 
\end{align}
where $\lambda_{TC}$ is the time-coordination error convergence rate. Consider
\\
\begin{adjustbox}{max width=\linewidth}
\parbox{\linewidth}{\begin{align} 
\label{UUU}
    &U-3\lambda_{TC}M_2 \nonumber \\
    &=
    \begin{bmatrix}
    \frac{a}{b}\lambda_{min}(\Xi)-3\lambda_{TC} \lambda_{max}(\Psi) & -\frac{1}{2}\left(\frac{2a}{b}n\lambda_{max}(\Psi)+\beta\frac{a}{b}n\right) \\ 
    -\frac{1}{2}\left(\frac{2a}{b}n\lambda_{max}(\Psi)+\beta\frac{a}{b}n\right) & \beta\left(b-\frac{a}{b}n-\frac{3}{2}\lambda_{TC}\right)
    \end{bmatrix}. 
\end{align}}
\end{adjustbox}
Note that for a fixed value of $\frac{a}{b}$, all the terms in \eqref{UUU} are fixed except for $\beta b$ in the $(2,2)$ element. Thus, choosing a sufficiently large $b$ with a fixed $\frac{a}{b}$ makes sure that \eqref{UUU} is positive semi-definite. 
Therefore, we obtain that $-z^\top Uz\leq-3\lambda_{TC}z^\top M_2z\leq-3\lambda_{TC}V_{TC}$, and thus the derivative of $V_{TC}$ is  bounded above by
\begin{adjustbox}{max width=\linewidth}
\parbox{\linewidth}{\begin{align*}
    \dot{V}_{TC}&\leq -3\lambda_{TC}V_{TC} \\
    &\mathrel{\phantom{\leq}}+\left(2\lambda_{max}(\Psi)+\beta\right)\|\bar{\xi}_{TC}\|\left(an\|e\|+k_{PF}\|e_{PF}\|+|\ddot{\gamma}_d|\right) \\
    &\leq-2\lambda_{TC}V_{TC}-\lambda_{TC}\min\{c_1,\beta/2\}\|\bar{\xi}_{TC}\|^2 \\
    &\mathrel{\phantom{\leq}}+\left(2\lambda_{max}(\Psi)+\beta\right)\|\bar{\xi}_{TC}\|\left(an\|e\|+k_{PF}\|e_{PF}\|+|\ddot{\gamma}_d|\right).
\end{align*}}
\end{adjustbox}
By applying Lemma $4.6$ in \cite{kha2002} and introducing the state transformation
$\bar{\xi}_{TC}=S\xi_{TC}\trieq
\begin{bmatrix}
b\mathbb{I}_{n-1} & Q \\ 
0 & \mathbb{I}_n
\end{bmatrix}\xi_{TC}$,~we~conclude that
\begin{equation} \label{iss}
\begin{aligned}
    \|\xi_{TC}(t)\|&\leq\kappa_1\|\xi_{TC}(0)\|e^{-\lambda_{TC}t} \\
    &\mathrel{\phantom{\leq}}+\kappa_2\sup_{t\geq0}\left(an\sqrt{n}h(t)+k_{PF}\|e_{PF}(t)\|+|\ddot{\gamma}_d(t)|\right), 
\end{aligned}
\end{equation}
where $\|e(t)\|\leq \sqrt{n}h(t)$ was used and
\begin{align}
    \kappa_1&\trieq\|S^{-1}\|\sqrt{\frac{\max\{c_2,\beta/2\}}{\min\{c1,\beta/2\}}}\|S\|, \label{kappa1} \\
    \kappa_2&\trieq\|S^{-1}\|\sqrt{\frac{\max\{c_2,\beta/2\}}{\min\{c_1,\beta/2\}}}\frac{\frac{k^2_\phi c_3}{\lambda}+\beta}{\lambda_{TC}\min\{c_1,\beta/2\}} \label{kappa2} 
\end{align}
with $c_1=\lambda_{min}(\Psi)$ and $c_2=\lambda_{max}(\Psi)$. \\
Lastly, we show that $t^i_{k+1}-t^i_k$ is  bounded below. From \eqref{estimator} and \eqref{dyn1'}, it follows that the estimation error dynamics can be written in the form of
\begin{equation}
\begin{aligned}
    \dot{\epsilon}_i(t)&=A\epsilon_i(t)+Bu_i(t), \ \ \ t\in[t^i_k,t^i_{k+1}), \label{dyn:error}\\
    \epsilon_i(t^i_k)&=[0 \ 0]^\top,
\end{aligned}
\end{equation}
where $\epsilon_i(t)=[e_i(t) \ \dot{e}_i(t)]^\top$, $A=\begin{bmatrix}
0 & 1 \\
0 & -b \\
\end{bmatrix}$, $B=\begin{bmatrix}
0 \\
1 \\
\end{bmatrix}$, and $u_i=a\sum_{j\in\mathcal{N}_i}(\gamma_i-\gamma_j)-a\sum_{j\in\mathcal{N}_i}e_j-\bar{\alpha}_i(e_{PF,i})$.

\noindent One can show that \\
\begin{adjustbox}{max width=\linewidth}
\parbox{\linewidth}{\begin{align*}
    |u_i|&\leq a\|L\gamma\|+a\|\mathcal{A}e\|+\|\bar{\alpha}(e_{PF})\| \\
    &\leq a\|L\|\|Q^\top\|\|Q\gamma\|+a\|\mathcal{A}\|\|e\|+k_{PF}\|e_{PF}\| \\
    &\leq an\left(\kappa_1\|\xi_{TC}(0)\|+\kappa_2\sup_{t\geq0}\left(an\sqrt{n}h(t)+k_{PF}\|e_{PF}\|+|\ddot{\gamma}_d|\right)\right) \\
    &\mathrel{\phantom{\leq}}+a\|\mathcal{A}\|\|e\|+k_{PF}\|e_{PF}\| \\
    &\leq an\kappa_1\|\xi_{TC}(0)\|+an\kappa_2\left(an\sqrt{n}\left(c_1+c_2\right)+k_{PF}\rho+\ddot{\gamma}_{d,max}\right) \\
    &\mathrel{\phantom{\leq}}+an\sqrt{n}(c_1+c_2)+k_{PF}\rho\trieq\bar{u}.
\end{align*}}
\end{adjustbox}
Therefore, we can show that the error $e_i(t)$ is bounded above using \eqref{dyn:error} as follows:
\begin{equation}
\begin{aligned}
    |e_i|\leq\|\epsilon_i\|&\leq\int^t_{t^i_k}(\|A\|\|\epsilon_i\|+\|B\||u_i|)d\tau \\
    &\leq\int^t_{t^i_k}\|A\|\|\epsilon_i\|d\tau+(t-t^i_k)\|B\|\bar{u}.
    \label{GB}
\end{aligned}
\end{equation}
Applying Gronwall-Bellman inequality (Lemma A.1 in \cite{kha2002}) to  inequality \eqref{GB} leads to \\
\begin{adjustbox}{max width=\linewidth}
\parbox{\linewidth}{\begin{align*}
    |e_i|&\leq\|\epsilon_i\|\leq(t-t^i_k)\|B\|\bar{u}+\|A\|\int^t_{t^i_k}(s-t^i_k)\|B\|\bar{u}e^{\|A\|(t-s)}ds \\
    &\leq\|B\|\bar{u}\left( e^{\|A\|(t-t^i_k)}-1\right)/\|A\|.
\end{align*}}
\end{adjustbox}
Considering that the transmission event is triggered when $|e_i(t)|>h(t)\geq c_1$, the inter-event time interval is bounded below:
\begin{align*}
    t^i_{k+1}-t^i_k\geq\frac{1}{\|A\|}\ln{(1+c_1\|A\|/(\|B\|\bar{u}))}>0.
\end{align*}
This completes the proof of Theorem~\ref{thm}.
\end{proof} 

\vspace{0.1em}

\begin{remark} \label{rem:conv_rate}
    Notice that the exponential convergence rate \eqref{convergence rate} depends on the ratio $\frac{a}{b}$ of the coordination control gains $a$ and $b$ in \eqref{dyn1}. A large value of $\frac{a}{b}$ leads to fast  inter-vehicle coordination.
\end{remark}

\section{SIMULATION RESULTS} \label{V}
This section presents simulation results of a coordinated path-following mission, which illustrate the efficacy of the proposed algorithm. The trajectory generation algorithm \cite{Cichella2021} designs a set of Bezier curves with the following specifications. The starting points are on $y=0\,m$. The trajectories simultaneously reach $y=150\,m$ exchanging $(x,z)$ coordinates. The inter-vehicle safety distance is $10\,m$. The mission duration is $t_f=21.10\,s$. The solid curves in Figure~\ref{fig:traj} depict the desired trajectories.
\begin{figure} [h!]
    \centering
    \includegraphics[width = 1.00\linewidth]{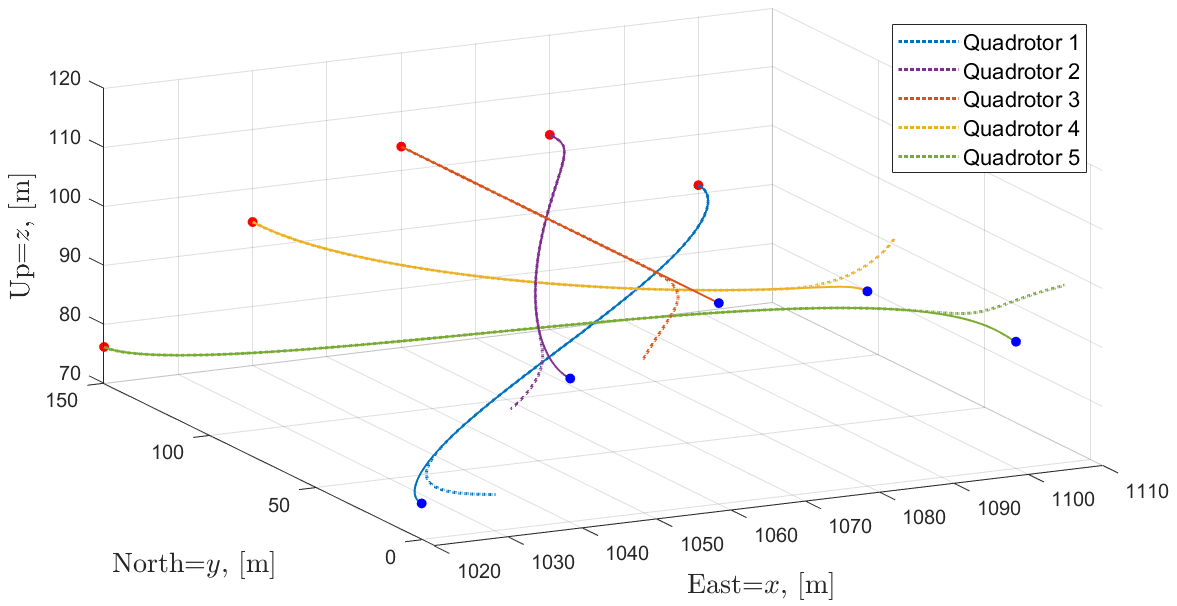}
    \hspace{0.4em}
    \caption{Time-coordinated path-following of five quadrotors. The starting points of the desired trajectories, blue dots, are on $y=0\,m$. The final points, red dots, are on $y=150\,m$.}
    \label{fig:traj}
\end{figure}

To emulate a situation where the quadrotors are hit and discoordinated by disturbances, the initial positions of the quadrotors are not the same as the initial points of the desired trajectories, that is, they have initial path-following errors. The path-following controller \cite{CICHELLA201313} steers each quadrotor to the desired trajectory. The dotted curves in Figure~\ref{fig:traj} are the paths traveled by the quadrotors.

Figure~\ref{fig:digraph} shows  the topology of the underlying communication network, which is connected satisfying Assumption~\ref{assum3}. The coordination control gains and the parameters in \eqref{alpha} are set to $a=3.75$, $b=4.82$, $k_{PF}=1.5$, and $\eta=12$. The initial conditions for the coordination states are $\gamma(0)=0_n$ and $\dot{\gamma}(0)=1_n$.
\begin{figure} [h!]
    \centering
    \includegraphics[width = 1.00\linewidth]{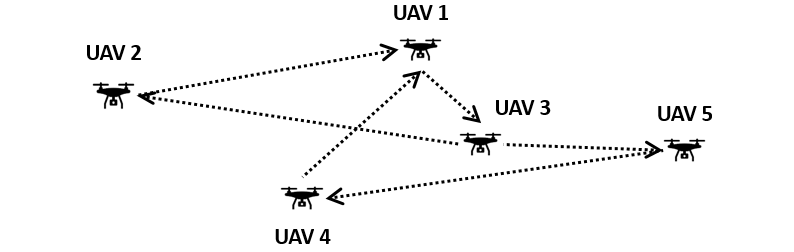}
    \caption{Topology of the underlying communication network.}
    \label{fig:digraph}
\end{figure}

Figures~\ref{fig:delta_gamma}-\ref{fig:event} illustrate how the coordination controller~\eqref{dyn1} works in combination with \eqref{estimator}. Initially, quadrotors $3$ and $5$ lie ahead of the plane $y=0\,m$, quadrotor $4$ lies behind it, and the remaining ones are on it. By the definition of~\eqref{alpha}, the term $\bar{\alpha}_i(e_{PF,i}(0))$ is positive for $i=3,5$; negative for $i=4$; and close to zero for the remaining $i$'s. As $\bar{\alpha}_i(e_{PF,i}(t))$ is added to the right hand side of \eqref{dyn1}, it makes sense that $\gamma_3(t)$, $\gamma_5(t)$ accelerate and $\gamma_4(t)$ decelerates for the first few seconds in Figure~\ref{fig:gamma_dot}. This evolution of $\gamma_i(t)$ helps each virtual target move towards the quadrotor and secure it inside the finite region of attraction of the path-following controller; however, in this process, the inter-vehicle coordination is destroyed. The recovery of it is guaranteed by the second term in \eqref{dyn1}. Figure~\ref{fig:delta_gamma} shows that the coordination is achieved in about $7\,s$. This coordination recovery time can be shortened by increasing the ratio $\frac{a}{b}$, as pointed out in Remark~\ref{rem:conv_rate}. Comparing \eqref{dyn1} with \eqref{estimator}, we can notice that the estimation error $e_i(t)=\hat{\gamma}_i(t)-\gamma_i(t)$ stems from $\bar{\alpha}_i(e_{PF,i}(t))$. With larger values of $\bar{\alpha}_i(e_{PF,i}(0))$ for $i=3,4,5$, the corresponding error $e_i(t)$, $i=3,4,5$ evolves fast and whenever it reaches the threshold function $h(t)$, sampling and transmission take place. In our simulation, we set $h(t)=0.03$. In Figure~\ref{fig:event}, it is confirmed that quadrotors $3$, $4$, and $5$ more frequently take a sample and transmit it right after the mission unfolds than the others do. The transmitted data is used by the estimator \eqref{estimator} to compute the estimate, which is then used in the second term of \eqref{dyn1} to achieve the inter-vehicle coordination. The fleet simultaneously arrives on $y=150\,m$ at $t=18.38\,s$, earlier than the original schedule $t_f=21.10\,s$. This is because we increased $\dot{\gamma}_d(t)$ from $1$ to $1.4$ in the middle of the mission. As the velocity of the virtual target is $\frac{dp_{d,i}(\gamma_i(t))}{d\gamma_i(t)}\dot{\gamma}_i(t)$, where $\dot{\gamma}_i(t)$ tracks $\dot{\gamma}_d(t)$, the increase of $\dot{\gamma}_d(t)$ led the fleet of quadrotors to increase the progression speed.

Notice that once the coordination errors become sufficiently small, the event is triggered less frequently, Figure~\ref{fig:event}. Therefore, we can say that the inter-vehicle communication occurs intelligently only when discoordination between the quadrotors increases. In summary, we confirmed that even though the fleet is hit by disturbances and blown away from the desired positions, the proposed ETC-based algorithm tries not to lose the UAVs by moving the virtual targets toward them and then recover the inter-vehicle coordination. All of these are done automatically without human intervention. Also, the communication takes place only when necessary, Figure~\ref{fig:event}. In other words, the proposed ETC-based algorithm achieves the coordination of UAVs with significantly reduced inter-vehicle communications compared to the authors' previous algorithms \cite{Kaminer2006,Ghabcheloo2007133,Xargay2013499,Cichella2015945}, \cite{Kang2024} based on continuous communications.

\begin{figure} [h!]
    \centering
    \includegraphics[width = 1.00\linewidth]{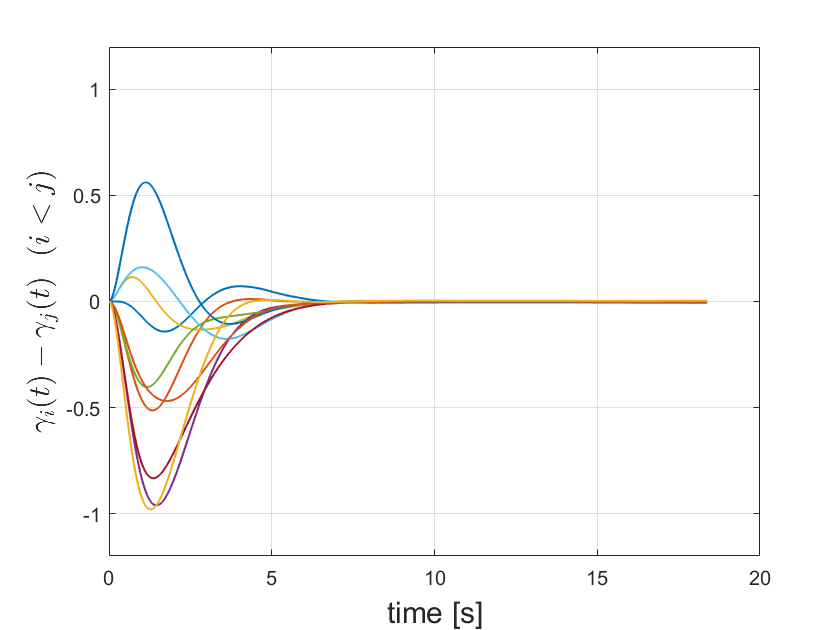}
    \caption{Convergence of the coordination error $\gamma_i(t)-\gamma_j(t)$ $(i<j)$ to a neighborhood of zero.}
    \label{fig:delta_gamma}
\end{figure}

\begin{figure} [h!]
    \centering
    \includegraphics[width = 1.00\linewidth]{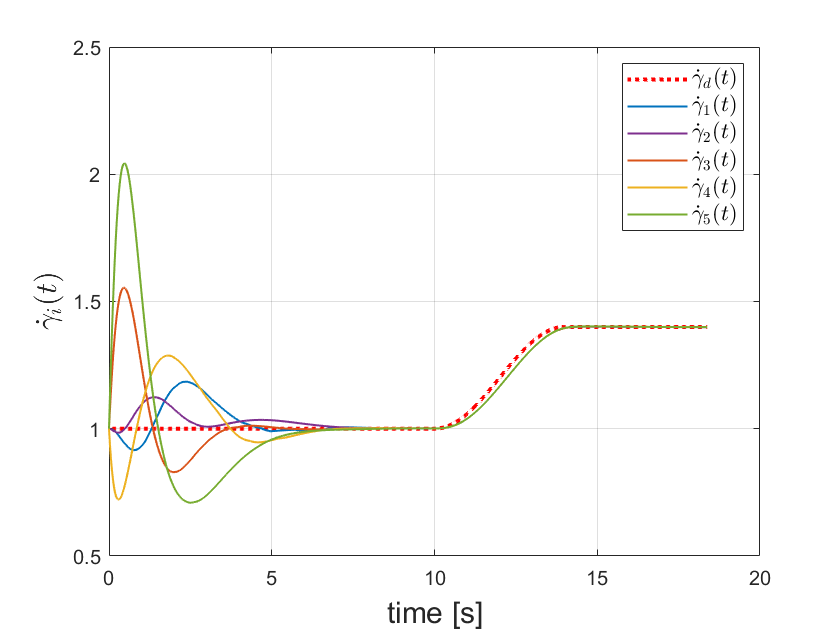}
    \caption{Convergence of $\dot{\gamma}_i(t)$ to a neighborhood of $\dot{\gamma}_d(t)$.}
    \label{fig:gamma_dot}
\end{figure}

\begin{figure} [h!]
    \centering
    \includegraphics[width = 1.00\linewidth]{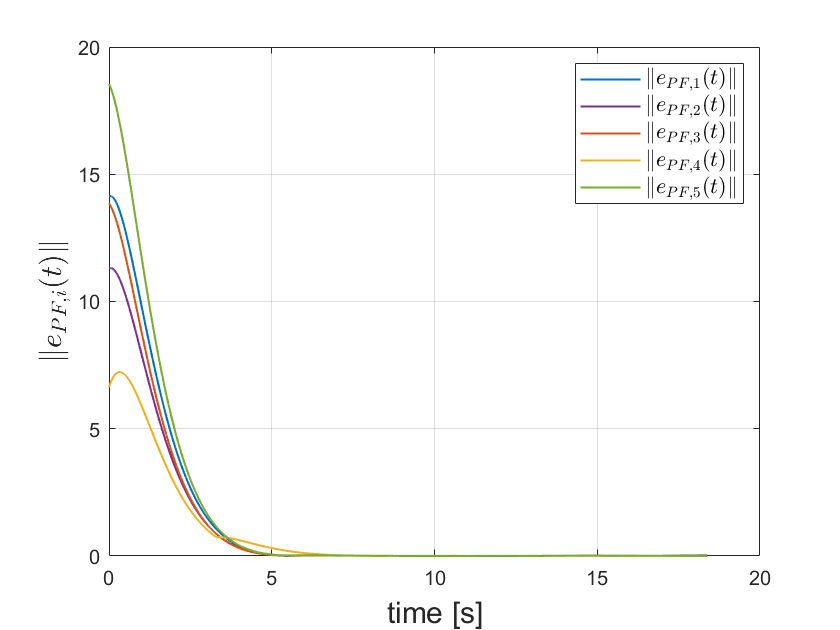}
    \caption{Path-following errors.}
    \label{fig:e_pf}
\end{figure}

\begin{figure} [h!]
    \centering
    \includegraphics[width = 1.00\linewidth]{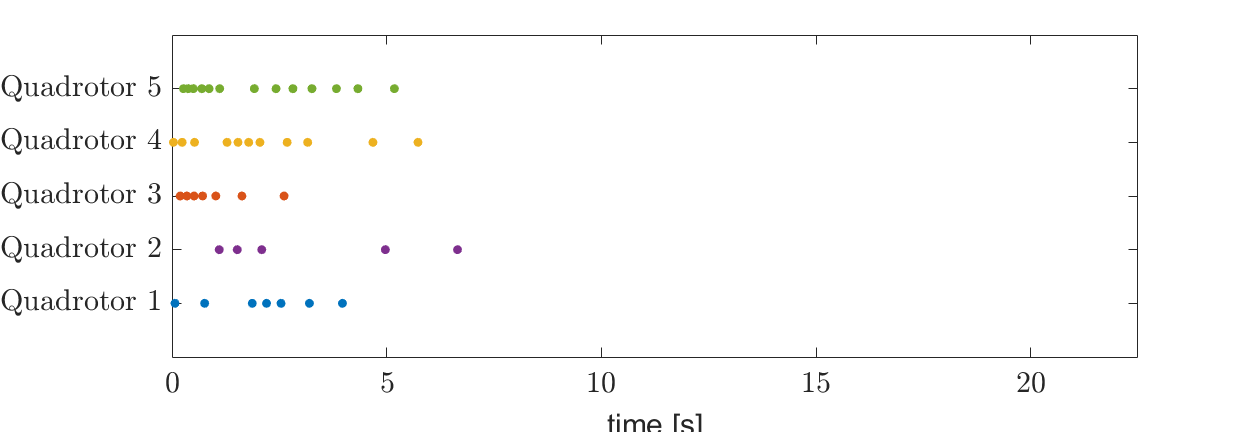}
    \caption{Event-triggered time instances.}
    \label{fig:event}
\end{figure}
\section{CONCLUSION} \label{VI}
This paper proposed a novel time-coordination algorithm using an event-triggered communication strategy. The exponential convergence of the coordination errors to a neighborhood of zero was proven using an ISS framework. Simulation results validated that the time coordinated path following of quadrotors can be achieved with the proposed algorithm. In future work, we will present an extension of the current ETC-based algorithm that works over a switching network. In a mission where a large number of UAVs is involved, the volume of required data transmission at a certain time might be above the capacity of the communication bandwidth. In order to circumvent this issue, it is more desirable to use a switching network instead of a static connected network. To be more specific, each UAV activates its transmission channels according to a predefined switching rule $\sigma(t)$ such that the network $\mathcal{D}_{\sigma(t)}$ is connected in an integral sense \cite{Kang2024}.



\bibliographystyle{ieeetr}
\bibliography{references}

\end{document}

%% file: definition.tex
\usepackage{adjustbox}
\usepackage{xcolor}
\usepackage{float}
\usepackage{amsmath}

\usepackage{amsthm,amssymb,graphicx,url}
\let\labelindent\relax 
\usepackage{enumitem}
\usepackage{booktabs}
\usepackage{caption}
\usepackage{subcaption}

\usepackage{appendix}

\usepackage{wrapfig}
\usepackage{hyperref}
\usepackage{soul}  
\usepackage{cleveref}
\usepackage{bm}
\usepackage{multicol}
\usepackage{mathtools}
\usepackage{cite}
\usepackage{graphicx}

\crefname{equation}{}{} 
\crefname{assumption}{Assumption}{}
\crefname{table}{Table}{} 
\crefname{figure}{Fig.}{}
\crefname{section}{Section}{}
\crefname{remark}{Remark}{}
\usepackage{algorithm,algorithmicx} 
\usepackage{algpseudocode} 
\newlength\myindent
\def\trieq{\triangleq}

\newtheorem{theorem}{Theorem}
\newtheorem{lemma}{Lemma}
\theoremstyle{definition}  
\theoremstyle{definition} \newtheorem{assumption}{Assumption}
\theoremstyle{remark}  \newtheorem{remark}{Remark}